\documentclass[11pt]{article}

\usepackage[margin=1in]{geometry}
\usepackage{color}
\usepackage{xcolor}
\usepackage{amssymb,amsmath,amsthm,amsfonts}
\usepackage{url}
\usepackage{hyperref}
\usepackage{framed}
\usepackage{bm}
\usepackage{enumerate}
\usepackage{multirow}

\newtheorem{theorem}{Theorem}
\numberwithin{theorem}{section}

\newtheorem{lemma}[theorem]{Lemma}
\newtheorem{corollary}[theorem]{Corollary}
\newtheorem{remark}[theorem]{Remark}

\theoremstyle{definition}
\newtheorem{definition}[theorem]{Definition}
\newtheorem{example}[theorem]{Example}

\newcommand{\CHI}{\hbox{\raise .4ex \hbox{$\chi$}}}

\newcommand{\R}{{\mathbb{R}}}
\newcommand{\Sb}{\mathbb{S}}

\newcommand{\E}{{\mathbb{E}}}

\newcommand{\Sbd}{\mathbb{S}^{d-1}}

\newcommand{\e}{\varepsilon}

\newcommand{\p}{\varphi}
\renewcommand{\phi}{\varphi}
\renewcommand{\epsilon}{\varepsilon}

\DeclareMathOperator{\conv}{conv}

\newcommand{\vd}{{\bf d}}
\newcommand{\ve}{{\bf e}}
\newcommand{\vf}{{\bf f}}
\newcommand{\vg}{{\bf g}}
\newcommand{\vh}{{\bf h}}

\newcommand{\vu}{{\bf u}}
\newcommand{\vv}{{\bf v}}

\newcommand{\vx}{{\bf x}}
\newcommand{\vy}{{\bf y}}
\newcommand{\vz}{{\bf z}}
\newcommand{\vA}{{\bf A}}

\newcommand{\vD}{{\bf D}}

\newcommand{\vF}{{\bf F}}

\newcommand{\vI}{{\bf I}}

\newcommand{\vX}{{\bf X}}
\newcommand{\vY}{{\bf Y}}

\newcommand{\vzero}{{\bf 0}}
\newcommand{\vPhi}{{\boldsymbol \Phi}}
\newcommand{\vSigma}{{\boldsymbol \Sigma}}
\newcommand{\valpha}{{\boldsymbol \alpha}}

\newcommand{\vmu}{{\boldsymbol \mu}}
\newcommand{\vphi}{{\boldsymbol \varphi}}
\newcommand{\vpsi}{{\boldsymbol \psi}}

\newcommand{\bB}{\mathbb{B}}

\newcommand{\bE}{\mathbb{E}}

\newcommand{\bS}{\mathbb{S}}

\begin{document}


\title{Preserving Injectivity under Subgaussian Mappings and Its Application to Compressed Sensing}

\author{Peter G. Casazza\footnote{Department of Mathematics, University of Missouri, Columbia, MO}\qquad Xuemei Chen\footnote{Department of Mathematical Sciences, New Mexico State University, Las Cruces, NM}\qquad Richard G. Lynch\footnote{Department of Mathematics, Texas A\&M University, College Station, TX}}

\maketitle

\begin{abstract}
	The field of compressed sensing has become a major tool in high-dimensional analysis, with the realization that vectors can be recovered from relatively very few linear measurements as long as the vectors lie in a low-dimensional structure, typically the vectors that are zero in most coordinates with respect to a basis. However, there are many applications where we instead want to recover vectors that are sparse with respect to a dictionary rather than a basis. That is, we assume the vectors are linear combinations of at most \(s\) columns of a \(d \times n \) matrix $\mathbf{D}$, where \(s \) is very small relative to \(n\) and the columns of $\mathbf{D}$ form a (typically overcomplete) spanning set.
	In this direction, we show that as a matrix $\mathbf{D}$ stays bounded away from zero in norm on a set $S$ and a provided map ${\boldsymbol \Phi}$ comprised of i.i.d. subgaussian rows has number of measurements at least proportional to the square of $w(\mathbf{D}S)$, the Gaussian width of the related set $\mathbf{D}S$, then with high probability the composition ${\boldsymbol \Phi} \mathbf{D}$ also stays bounded away from zero. 
As a specific application, we obtain that the null space property of order \(s\) is preserved under such subgaussian maps with high probability. Consequently, we obtain stable recovery guarantees for dictionary-sparse signals via the $\ell_1$-synthesis method with only $O(s\log(n/s))$ random measurements and a minimal condition on $\mathbf{D}$, which complements the compressed sensing literature.

\end{abstract}

%
\section{Introduction}\label{sec:intro}

An important problem in high-dimensional analysis is to recover a signal from undersampled and corrupted measurements. This problem is ill-posed if no further assumptions are imposed on the signal class. With the breakthrough of compressive sensing (CS)
(see \cite{FRbook}), we now know that it is possible to recover signals from very few (typically noisy) measurements, provided that the signals are sitting in a low-dimensional structure. 

To make this more concrete, we will use standard CS terminology. We wish to recover a signal $\vz_0\in\R^d$ from its undersampled and corrupted linear measurements $\vy=\vPhi  \vz_0+\ve\in\R^m$, with the noise satisfying $\|\ve\|_2\leq\epsilon$. 
The number of measurements $m$ is assumed to be far less than the ambient dimension $d$, meaning the system has infinitely many solutions in general. To surpass this hurdle, we assume the signal $\vz_0$ has a sparse structure, that is, it can be written as the linear combination of only a few atoms from a dictionary. In other words, if $\vD$ is the matrix whose columns are the atoms, then $\vz_0=\vD\vx_0$ for some sparse vector $\vx_0$. 

To recover dictionary-sparse signals, there are mainly two classes of algorithms: convex programming \cite{CDS99, T96, CRT06, FL09} and greedy algorithms \cite{Z11, F11, CDD17, NT08}. This paper will  focus on convex problems of the following form:
\begin{equation}\label{equ:fDe}
\widehat{\vz}=\arg\min_{\vz\in \R^d} f_\vD(\vz), \quad\text{ subject to } \|\vy-\vPhi \vz\|_2\leq \epsilon,
\end{equation}
where $f_\vD(\vz)$ is a convex function of $\vz$. The unified form \eqref{equ:fDe} is not new, see \cite{CRPW12} for example. Two questions of interest when attacking this problem are the following:
\begin{itemize}
	\item[(Q1)] What dictionaries ensure the existence of a sensing matrix $\vPhi$ to recover $\vz_0$ stably from \eqref{equ:fDe}?
	\item[(Q2)] Given such a dictionary \(\vD\), how  do we find suitable sensing matrices with number of rows as small as possible?
\end{itemize}
We begin by introducing in the next two subsections some currently known results related to answering these fundamental questions.

\subsection{The Basis Case}
Most CS literature focuses on the case when $\vD$ is the canonical basis, i.e., $\vz_0=\vx_0$. 
A suitable sensing matrix would mean that $\vPhi$ is able to extract the low-dimensional information off of $\vz_0$. This is reflected in popular conditions like low coherence, the restricted isometry property, and the null space property, as well as their variations. In this case, the most popular method is the $\ell_1$-minimization, also known as the  Basis Pursuit \cite{CRT06,FL09}, where $f_\vD(\vz)=\|\vz\|_1$:
\begin{equation}\label{equ:Pe}
\widehat{\vz}=\arg\min_{\vz\in \R^d} \|\vz\|_1 \quad\text{ subject to } \quad \|\vy-\vPhi \vz\|_2\leq \epsilon.
\end{equation}

Proposed by Candes and Tao \cite{CT05, CT06}, the \emph{restricted isometry property (RIP)} is sufficient to recover sparse signals via \eqref{equ:Pe}. A matrix \(\vPhi \in \mathbb{R}^{m \times d} \) satisfies the RIP with constant \(0 \leq \delta < 1\) and sparsity \(s\) if 
\[
(1 - \delta) \|\vz\|_2^2 \leq \|\vPhi\vz \|^2_2 \leq (1 + \delta) \|\vz\|_2^2
\]
for all \(s\)-sparse signals \(\vz \in \mathbb{R}^d \). The smallest \(\delta_s \geq 0 \) for which the RIP holds is called the \emph{restricted isometry constant}. This condition 
ensures that distinct sparse vectors have sufficiently far away measurements, providing explicit recovery guarantees. 

While the RIP is a sufficient condition for recovery guarantees via \eqref{equ:Pe}, another property known as the \emph{null space property (NSP)} is both necessary and sufficient. A sensing matrix \(\vPhi \in \mathbb{R}^{m \times d} \)  is said to satisfy the null space property of order $s$ if for any index set $T$ with $|T|\leq s$ and any $\vz\in\ker(\vPhi)\backslash\{0\}$, 
$
\|\vz_T\|_1<\|\vz_{T^C}\|_1
$
holds. Here \(\vz_T \) denotes the vector having the same entries as \(\vz \) on the {support} \(T\) and zero elsewhere. It is known that the successful recovery of sparse vectors from Basis Pursuit \eqref{equ:Pe} when $\epsilon=0$  occurs if and only if the NSP holds \cite{DE03, GN03}. Moreover, it was shown in \cite{ACP11} that the NSP is necessary and sufficient for the stable recovery via Basis Pursuit.  Other than the characterization of Basis Pursuit, another advantage of this property is that it only depends on the kernel of $\vPhi$, which means that this property is invariant under linear combinations of measurements (the rows of $\vPhi$). By a compactness argument, the NSP of order $s$ is equivalent to the existence of $0<\gamma<1$  such that
\begin{equation}\label{equ:snsp}
 \|\vz_T\|_1<\gamma\|\vz_{T^C}\|_1, \text{ for all } \vz\in\ker(\vPhi)\backslash\{0\}.
\end{equation}
This is the so-called \emph{stable null space property}.

The stable NSP is a matrix property.  However, we will abuse the notation and say \emph{a vector $\vz$ has the stable NSP} if $\|\vz_T\|_1<\gamma\|\vz_{T^C}\|_1$ for any index set $T$ with cardinality at most $s$. We let \(S_\gamma^d\) be the set of vectors on the unit ball of $\R^d$ that do not have the stable NSP. Explicitly, 
\begin{align}\label{equ:S2}
S_\gamma^d
:= \{\vx : \|\vx_T\|_1 \geq \gamma\|\vx_{T^c}\|_1 \mbox{ for some } |T| \leq s \} \cap \bS^{d- 1}_2 .
\end{align}

Taking the intersection with the unit ball in the definition of \(S_\gamma^d\) is mainly for convenience the stable NSP for $\vPhi$ is equivalent to the existence of a positive lower bound of $\|\vPhi\vx\|_2$ on $S_\gamma^d$. In summary, $\vPhi$ having the NSP is equivalent to the existence of $0<\gamma<1$ and $\eta>0$ such that
\begin{equation}\label{equ:srnsp}
 \inf\left\{\|\vPhi\vx\|_2 : \vx \in S_\gamma^d\right\} \geq \eta.
\end{equation}

\noindent We will be using the notation $S_\gamma^n$ frequently, so we conveniently denote it by $S_\gamma$ instead since the dimension should be clear from context.

Next, we provide a recovery result using \eqref{equ:srnsp} that we are not able find in the literature.

\begin{theorem}\label{thm:snsp}
If a sensing matrix $\vA\in\R^{m\times n}$ satisfies the stable NSP of order $s$ with 
$$\inf\left\{\|\vA\vv\|_2 : \vv\in S_\gamma\right\}\geq\eta \text{\quad for some }\eta> 0, $$
then given $\vy=\vA\vx+\ve$ with $\|\ve\|_2\leq \epsilon$, we have
 \begin{equation}\label{equ:snspresult}
 \|\widehat \vx-\vx\|_2\leq\frac{2\gamma+2}{1-\gamma}\sigma_s(\vx)+\frac{2\epsilon}{\eta} ,
 \end{equation}
 where 
 \begin{equation*}
\widehat{\vx}=\arg\min_{\vx\in \R^n} \|\vx\|_1 \quad\text{ subject to } \quad \|\vy-\vA \vx\|_2\leq \epsilon.
\end{equation*}
\end{theorem}

The condition \eqref{equ:srnsp} is similar to the so called robust null space property \cite[Definition 4.17]{FRbook}, and it resembles the recovery result of the robust NSP.  
We include its proof in Section \ref{sec:stab_nsp} for completeness.
The argument is fairly standard. In fact, a very similar argument can be found in the proof of \cite[Theorem 3]{DLR18}, which uses the robust null space property. 


The best answer to (Q2) so far in the basis case is to use random matrices as the sensing matrix \(\vPhi \). It is well-known that random matrices whose entries are drawn from Gaussian or Bernoulli random variable satisfy the RIP with high probability, provided that $m$  is only on the order of $s\log (d/s)$ \cite{CT05, RV08, DPW09}. On the other hand, one needs at least $O(s\log (d/s))$ number of measurements to ensure recovery, regardless of the decoder~\cite{FSurvey},  and therefore random constructions achieve this minimum. There are many other types of random matrices that recover the signal effectively with \eqref{equ:Pe}, but do not necessarily have RIP, the Weibull matrices \cite{F14} for example. This is evidence that RIP is stronger than the NSP, which is explicitly proven in~\cite{CCW16} using a semi-deterministic construction. The deterministic construction of suitable sensing matrices is significantly harder, and it requires many more measurements \cite{BDFKK11}.


Much more can be said about both the history and the theory of CS in the basis case. For those interested in learning more, see the book \cite{FRbook} and the survey \cite{FSurvey}.

\subsection{The General Dictionary Case}

The general setting $\vz_0=\vD\vx_0$ where \(\vD\) is an arbitrary full rank \(d \times n \) matrix is much more challenging. When \(d = n\), the columns of \(\vD\) form a basis for \(\mathbb{R}^d \) and it is not hard to see that we can translate this to the canonical basis case as described before. However, the difficulty occurs when we assume that \(n > d\) so the dictionary \(\vD\) is overcomplete. In this case, $\vz_0$ has infinitely many representations in $\vD$, including possibly more than one sparse representation. There are many applications where the signals are seen through such a transformation and the need to understand when stable recovery is achievable is immense \cite{RSV08,CENR11,LML12,ACP12,DNW12,CWW14,F16,KNW16}.  

We note that such an overcomplete dictionary is also often called a \emph{finite frame for \(\mathbb{R}^d\)}. The field of finite frame theory is rich and has proven to be a powerful asset in many modern, real-world applications. We refer the inquisitive readers to \cite{CK13, CL16} for a more thorough introduction to the elements of finite frame theory.

Perhaps the most reasonable recovery problem to consider in the dictionary setting, since it is the natural extension of \eqref{equ:Pe}, is the $\ell_1$-synthesis method:
\begin{equation}\label{equ:l1e}
\begin{array}{l}
\widehat \vx=\arg\min\|\vx\|_1 \quad\text{ subject to } \quad \|\vy - \vPhi \vD\vx\|_2\leq\epsilon\\
\widehat \vz=\vD\widehat \vx.
\end{array}
\end{equation}
We note that defining $\|\vz\|_{K_\vD}:=\min\{\|\valpha\|_1 : \vD\valpha=\vz\}$ gives the following reformulation of the $\ell_1$-synthesis method \eqref{equ:l1e}:
\begin{equation}\label{equ:K}
\widehat{\vz}=\arg\min_{\vz\in \R^d} \|\vz\|_{K_\vD} \quad\text{ subject to } \quad \|\vy-\vPhi \vz\|_2\leq \epsilon.
\end{equation}
Specifically, for any convex set $K$, the \emph{Minkowski functional} of $K$ is defined as $\|\vv\|_K:=\inf\{\lambda>0: \lambda^{-1}\vv\in K\}$ so that in the dictionary setting where $\vD=[\vd_1, \cdots, \vd_n]$ with $K_\vD:=\text{conv}\{\pm \vd_i\}_{i=1}^N$, we have $\|\vz\|_{K_\vD}=\min\{\|\valpha\|_1 : \vD\valpha=\vz\}$  \cite{V15}. The Minkowski functional is also known as the gauge of $K$, or the atomic norm associated to $K$.

One way to guarantee the successful recovery of \eqref{equ:l1e} is to require $\vPhi \vD$ to have the NSP or the RIP. The work by Rauhut et al. \cite{RSV08} showed that if $\vPhi\in\R^{m\times d}$ is a random matrix satisfying a concentration inequality with $m=O(s\log\frac{n}{s})$ and $\vD$ satisfies the RIP, then the matrix $\vPhi \vD$ also satisfies the RIP.


Once the composition $\vPhi \vD$ satisfies the RIP, the program \eqref{equ:l1e} (or \eqref{equ:K}) will stably recover the sparse representation $\vx_0$, and consequently the signal $\vz_0$. However, we often only care about the recovery of $\vz_0$ in this dictionary based sparsity problem, in which case we allow $\widehat{\vx}$ to be far away from $\vx_0$.

To approach the problem in this new light, the work in \cite{CENR11} instead proposed the model where $f_\vD(\vz)=\|\vD^*\vz\|_1$ in \eqref{equ:fDe}, called the $\ell_1$-analysis  method:
\begin{equation}\label{equ_ana}
\widehat{\vz}=\arg\min_{\vz\in \R^d} \|\vD^*\vz\|_1 \quad\text{ subject to } \quad \|\vy-\vPhi \vz\|_2\leq \epsilon.
\end{equation}
They showed that successful recovery via \eqref{equ_ana} is possible when \(\vD \) is a Parseval frame, i.e. \(\vD\vD^* =  \vI_{d} \) and provided that \(\vPhi \) satisfies a dictionary based RIP, $\vD$-RIP. The definition of $\vD$-RIP is similar to the usual RIP, but with \(\vD\vx \) in place of \(\vx \).


The $\ell_1$-analysis and $\ell_1$-synthesis models assume different sparsity to begin with. The analysis model assumes the sparsity of the  analysis coefficient $\vD^*\vz$, which has applications in imaging where $\vD$ can be the finite difference operator, wavelets, shearlets, etc. \cite{KKZ14, KR15}. The $\ell_1$-synthesis model assumes one of the infinitely many coefficients for $\vz$ in $\vD$ is sparse, as introduced at the beginning. This is more inclusive as the analysis coefficient is a particular case where the dual frame is the analysis operator ($\vz=\vD\vD^*\vz$), see \cite{LML12} for more details.  On the technical side, the synthesis approach often imposes more challenges due to its setting, and the fact that we do not know which dual frame of $\vD$ generates a sparse representation.
The work by Chen et al. \cite{CWW14} tackled the $\ell_1$-synthesis problem and aimed to lay a framework for this method.  They proposed a dictionary based NSP for the sensing matrix, \(\vD\)-NSP for short, which we now define.

\begin{definition}
	Let \(\vD \in \R^{d \times n} \) be a dictionary. A matrix \(\vPhi \in \R^{m \times d} \) is said to satisfy the \(\vD \)-NSP of order \(s\) if for any index set \(T\) with \(|T| \le s \) and any $\vv$ such that \( \vD \vv\in \ker\vPhi \backslash \{\vzero \}  \), there is some \(\vu \in \ker \vD \) so that
	\[
	\|\vv_T + \vu \|_1 < \|\vv_{T^c}\|_1.
	\]
\end{definition}

The \(\vD\)-NSP  is a characterization of exact recovery of dictionary sparse signals via \eqref{equ:l1e} when $\epsilon=0$, and therefore is a generalization of the NSP. The following result helps emphasize the general direction of this paper.


\begin{theorem}[{\cite[Theorem 7.2]{CWW14}}]\label{thm:DNSP}
	If $\vD$ is full spark, then $\vPhi$ has the \(\vD\)-NSP with sparsity $s$ if and only if $\vPhi \vD$ has the NSP with the same sparsity.
\end{theorem}

A frame $\vD\in\R^{d\times n}$ is \emph{full spark} if every collection of $d$ frame vectors is linearly independent. Full spark is not a strong assumption on dictionaries. In fact, it is quite obvious that if we randomly choose the entries of $\vD$ according to any continuous distribution, then \(\vD\) will be full spark with probability one. More details can be found in \cite{ACM12}.

As a (surprising) result of Theorem \ref{thm:DNSP},  if the $\ell_1$-synthesis method is successful at all, almost always, we will recover both $\vx_0$ and $\vz_0$, and $\vD$ will satisfy the NSP since $\ker(\vD)\subset\ker(\vPhi \vD)$. In other words, if we are using $\ell_1$-synthesis to recover $\vz_0$, then it is very reasonable to assume that $\vD$ has NSP and the coefficients $\vx_0$ will be recovered simultaneously. Therefore we will study the properties of the composition $\vPhi\vD$ to ensure the success of $\ell_1$-synthesis.

Like the basis case, most work for the dictionary case often uses random measurements. The paper~\cite{KR15} uses Gaussian measurements for the $\ell_1$-analysis method, providing both nonuniform and uniform guarantees. The work \cite{F14} also considers the \(\ell_1 \)-analysis approach, but instead uses Weibull measurements. 
As mentioned earlier, the work by Rauhut et al. \cite{RSV08} does analyze the $\ell_1$-synthesis method, however, it requires the dictionary $\vD$ to have the RIP. We again note that there is a gap between the RIP and the NSP~\cite{CCW16}, so we would like to reduce this assumption on $\vD$. 
Additionally, it is known that a subgaussian matrix \(\vPhi \) satisfies with high probability the \(\vD \)-RIP~\cite{CENR11}, from which it is not hard to see that if \(\vPhi \) is subgaussian, then essentially (with high probability) \(\vPhi \vD \) has the RIP if and only if \(\vD \) has the RIP. This further solidifies the need to weaken the RIP assumption.
Another notable work on random measurements in this setting is by Vershynin~\cite{V15}, which directly measures the recovery error in expectation. However, the error bound does not promote sparsity, and therefore will not provide exact reconstruction for $s$-sparse representations. 

This paper will study the recovery prospects of \eqref{equ:l1e} or \eqref{equ:K} when the measurements are chosen at random. This kind of model is used in data acquisition when random measurements can be extracted, and is also applied in machine learning where data are assumed to have certain distributions. In this paper, we will assume the measurements are subgaussian, which we will review in Section \ref{sec:sg}.  We wish to justify the effectiveness of random sensing matrices for recovering dictionary-sparse signals when $\ell_1$-synthesis \eqref{equ:l1e} is used. As explained, this reduces to verifying the null space property of the composite $\vPhi\vD$. The biggest question is how small the number of measurements $m$ can be given $n, d, s$ fixed. If we treat $\vPhi\vD$ as a whole and focus on recovering the coefficient $\vx_0$, then this reduces to the basis case and the optimal number of measurements is $m=O(s\log\frac{n}{s})$, which is usually achieved by random construction as is the case in \cite{F14, KR15}. So we wish to answer
the question:
\begin{center}
\begin{tabular}{ll}
\emph{\quad Given that $\vD$ has the NSP, find the smallest number of measurements such that }&\multirow{2}{*}{\qquad\quad\quad(*)}\\
$\vPhi\vD$ \emph{has the NSP when the rows of $\vPhi$ are subgaussian.}&
\end{tabular}
\end{center}

\vspace{0.1in}

\noindent Notice that the NSP requirement on $\vD$ makes sense and is inevitable since $\ker(\vD)\subset\ker(\vPhi\vD)$.

\subsection{Our Contribution and Organization}

The contribution of this paper is twofold. The first main result, Theorem \ref{thm:S}, states that a certain property of an operator/dictionary can be preserved under a subgaussian random map, given that this map is projecting to a dimension that is on the order of the square of the Gaussian width of certain set. It is not a coincidence that we study this preservability since the problem we wish to solve has this flavor. However, this could potentially be used to analyze other properties of compressed sensing matrices, or even beyond the scope of sparse analysis. Our second main result is the application of Theorem \ref{thm:S} to the null space property, thus solving (*). Specifically, Theorem \ref{thm:main} states if \(\vPhi \in \mathbb{R}^{m\times d} \) is a sensing matrix with independently drawn rows from a subgaussian distribution, \(\vD \in \mathbb{R}^{d \times n} \) satisfies the NSP, and the number of measurements \(m\) is on the order of $s\log(n/s)$, then \(\vPhi \vD \) also satisfies the NSP with high probability.  Consequently, we get a recovery result stated in Corollary \ref{cor:main}, which is the first recovery result with only $s\log(n/s)$ subgaussian measurements that only requires the dictionary to be NSP. 
As far as the authors can tell, this is the first work on the $\ell_1$-synthesis algorithm with random measurements that only requires a minimal condition on $\vD$.

The road map is as follows. In Section \ref{sec:prelims}, we provide the required preliminary material and notations that will be used throughout the paper. In Section \ref{sec:main}, we introduce our main results as described above and argue how they are essentially optimal. Furthermore, we obtain as a corollary a more suitable recovery guarantee for the \(\ell_1 \)-synthesis method.
We then provide the theory behind our results in Section \ref{sec:proofs}, as well as alternative estimates for the Gaussian width.
 The more specific cases that \(\vPhi \) has rows drawn independent from a multivariate normal distribution is addressed in Section \ref{sec:gauss}, wherein we provide better estimates. Finally, we end  with some conclusions and discussions.
 

\section{Preliminaries}\label{sec:prelims}

We use $\|\cdot\|_p$ for the standard $\ell_p$ norm and we let 
$$\bS^{n-1}_p :=\{\vx\in\R^n :  \|\vx\|_p=1\}  \text{ and } {\bB^{n-1}_p :=\{\vx\in\R^n :  \|\vx\|_p\leq1\}}.$$ We also denote $[n]:=\{1,\cdots,n\}$. If \(S \subset \R^n \) and \(\vD \in \R^{d \times n} \) is a dictionary, then we write \(\vD S \) for \(\vD S = \{\vD \vx : \vx \in S \} \). Also, we denote the columns of \(\vD \) by $\vd_i$ so that $\vD=[\vd_1, \cdots, \vd_n]$. 
{For a dictionary
$\vD$, $\|\vD\|_2$ is the operator norm.}

We will use the notation \(X \sim N(\mu,\sigma^2) \) to mean that a one dimensional random variable \(X\) follows a normal distribution with mean \(\mu \) and variance \(\sigma^2 \), and \(\vX \sim N(\vmu,\vSigma) \) means a multidimensional random variable \(\vX \) follows a multivariate normal distribution with mean vector \(\vmu \) and covariance matrix \(\vSigma \).

\subsection{The Gaussian Width}
In the proof of our main result, we will need to bound \(w(\vD S_\gamma)\), where \(w \) denotes the Gaussian width defined as follows.
\begin{definition}
	The \emph{Gaussian width} of a set $S\subset\R^n$ is defined as
	\[
	w(S) := \bE \sup_{\vx \in S} \langle \vg , \vx \rangle,
	\]
	where $\vg \sim N(\vzero, \vI_n)$ is a standard Gaussian random vector.
\end{definition}
The Gaussian width plays a central role in asymptotic convex geometry. In particular, thinking of each inner product \(\langle \vg , \vx \rangle\) as a random projection, the Gaussian width measures how well, on average, the vectors in \(S\) can line up with a randomly chosen direction. 
{For example, the Gaussian width of the unit ball is $w(\bB_2^{n-1})=\E\|\vg\|_2=\sqrt{2}\Gamma(\frac{n+1}{2})/\Gamma(\frac{n}{2})$, which is on the order of $\sqrt{n}$.} It is in this way that the Gaussian width can be thought of as a way to measure the ``size'' of a set~\cite{W15}. In terms of CS, bounding the Gaussian width is how one obtains the important concentration equality used in the now standard CS proofs \cite[Chapter 9]{FRbook}. Therefore, it is natural that our proof techniques will make use of it as well. Lastly, we note that it is often required that the set $S$ be symmetric about the origin, which \(S_\gamma \) satisfies. 

We will need the following result. The argument is given on Page 10 of \cite{PVY17}, but we will provide it in Section \ref{sec:gw2} for the sake of completeness. 
\begin{lemma}\label{lem:wFS}
For any map $\vF \in \R^{d \times n}$ and any $S\subset\R^n$, we have
$$w(\vF S)\leq \|\vF\|_2 w(S).$$
\end{lemma}


\subsection{Subgaussian Vectors}\label{sec:sg}

The measurement matrix \(\vPhi \in \R^{m \times d} \) in our result will have rows drawn i.i.d. from a subgaussian distribution, which we now define following \cite{T14}.

\begin{definition} A random vector $\vphi\in\R^d$ is called a \emph{subgaussian vector} with parameters $(\alpha, \sigma)$ if it satisfies the following.
	\begin{enumerate}[{(1)}]
		\item It is centered, that is, $\E [\vphi]= \vzero$.
		\item There exists a positive $\alpha$ such that $\E\left[|\langle\vphi, \vz\rangle|\right]\geq\alpha$ for every $\vz\in\Sbd_2$.
		\item There exists a positive $\sigma$ such that $\Pr\left(|\langle\vphi,\vz\rangle|\geq t\right)\leq 2\exp(-t^2/(2\sigma^2))$ for every $\vz\in\Sbd_2$.
	\end{enumerate}
\end{definition}


There are many examples of subgaussian vectors, including vectors with independent Gaussian entries, or independent Bernoulli entries, as well as independent bounded entries. We list the following two cases in detail since they are used in corresponding theorems in the next section.
\begin{example}[Standard Gaussian vector]\label{ex:gau}
	Let $\vphi \in \R^{d} \sim N(\vzero, \vI_d)$ be a standard Gaussian vector. If \(\vz \in \Sbd_2 \), then
	$Z :=\langle \vphi, \vz\rangle \sim N(0,1)$ and it is well known that
	\begin{align*}
	\E [|Z|] = \sqrt{\dfrac{2}{\pi}} \qquad \mbox{and} \qquad
	\Pr\left(|Z|\geq t\right) \leq \exp\left(-\dfrac{t^2}{2}\right),
	\end{align*}
	so the standard Gaussian vector is subgaussian with parameters $\alpha= \sqrt{2/\pi}$ and $\sigma=1$.
\end{example}

\begin{example}[Nonstandard Gaussian vector]\label{ex:non}
	Suppose that $\vphi \in \R^{d} \sim N(\vzero,\vSigma)$  where the covariance matrix $\vSigma$ has smallest and largest singular values, $\sigma_{\min}^2$ and $\sigma_{\max}^2$, respectively. Then $\vpsi :=\vSigma^{-1/2}\vphi \sim N(\vzero, \vI_d)$ and we can compute for all \(\vz \in \Sbd_2 \) that
	\begin{align*}
	\E\left[|\langle\vphi, \vz\rangle|\right]
	= \E\left[|\langle\vSigma^{1/2}\vpsi, \vz \rangle|\right]=\E\left[|\langle\vpsi, \vSigma^{1/2}\vz\rangle|\right]=
	\|\vSigma^{1/2}\vz\|_2 \,\E\left[\left|\left\langle\vpsi, \frac{\vSigma^{1/2}\vz}{\|\vSigma^{1/2}\vz\|}\right\rangle\right|\right]\geq\sigma_{\min}\sqrt{\dfrac{2}{\pi}}
	\end{align*}
	and in a similar fashion
	\begin{align*}
	\Pr\left[|\langle\vphi, \vz\rangle|\geq t\right]
	=\Pr\left[\left|\left\langle\vpsi, \frac{\vSigma^{1/2}\vz}{\|\vSigma^{1/2}\vz\|}\right\rangle\right|\geq \frac{t}{\|\vSigma^{1/2}\vz\|_2}\right]
	\leq \exp\left(-\dfrac{t^2}{2\sigma_{\max}^2}\right),
	\end{align*}
	so that $\vphi$ is subgaussian with parameters $\alpha=\sigma_{\min}\sqrt{2/\pi}$ and $\sigma=\sigma_{\max}$.
	
\end{example}

\subsection{The Mean Empirical Width}
If $\{\vf_i\}_{i=1}^m$ are independent copies of the random distribution $\vf \in \R^n$, then we can define the \emph{mean empirical width} of a set $S \subset \R^n$ as
$$W_m(S; \, \vf) := \mathbb{E} \sup_{\vx \in S} \bigg\langle \vx, \dfrac{1}{\sqrt{m}} \sum_{i = 1}^m \e_i \vf_i \bigg\rangle,$$
	where $\{\e_i\}_{i = 1}^m$ are independent random variables taking values uniformly over $\{\pm 1\}$ and are independent from everything else. 

The mean empirical width $W_m(S;\, \vf)$ is a distribution-dependent measure of the size of the set $S$. 
Note that $W_m(S;\, \vf)$ reduces to the usual Gaussian width $w(S)$ when $\vf$ follows a standard Gaussian distribution, as shown in Remark \ref{rem:C}.
Estimation of \(W_m(S;\, \vf)\) for any subgaussian vector \(\vf\) is made in \cite{T14}, where \(S\) is required to be $\Sb_2^{n-1}\cup K$ for some cone $K$. However, the bound can be relaxed to any subset \(S\) by the observation of the generic chaining bound and the majorizing measure theorem \cite[{Theorem 2.2.18 and Theorem 2.4.1}]{Ta14}. We will state this as a lemma.

\begin{lemma}\label{lem:W}
	If $\vf \in \R^n$ is subgaussian with parameters $(\alpha, \sigma)$ and $S$ is {\bfseries any} subset of $\R^n$, then
	\begin{equation}
	W_m(S;\, \vf)\leq C\sigma w(S)
	\end{equation}
	for some universal constant $C$.
\end{lemma}

The constant $C$ will appear in the main results. It is a universal constant that does not rely on the choice of subgaussian distribution. See \cite{Ta14} for precise computations of this constant. However, better estimates can be made using different techniques as is shown in the next example.
\begin{remark}\label{rem:C}
When $\vf$ follows a centered multivariate normal distribution $N(\vzero, \vSigma)$, we can take $C=1$. Let $\vg_i:=\vSigma^{-1/2}\vf_i\sim N(\vzero, \vI_n)$, then $\vg=\sum_{i = 1}^m \e_i \vg_i\sim N(\vzero, \vI_n)$ as well, so
\begin{align*}
W_m(S; \vf)&=\E\sup_{\vz\in S}\bigg\langle \vz, \dfrac{1}{\sqrt{m}} \sum_{i = 1}^m \e_i \vf_i \bigg\rangle=\E\sup_{\vz\in S}\bigg\langle \vz, \dfrac{1}{\sqrt{m}} \sum_{i = 1}^m \e_i \vSigma^{1/2}\vg_i \bigg\rangle\\
&=\E\sup_{\vz\in S}\bigg\langle \vSigma^{1/2}\vz, \dfrac{1}{\sqrt{m}} \sum_{i = 1}^m \e_i \vg_i \bigg\rangle=\E\sup_{\vx\in \vSigma^{1/2}S}\langle \vx, g \rangle\\
&=w(\vSigma^{1/2}S)\leq \sigma_{\max}w(S),
\end{align*}
where the last inequality is due to Lemma \ref{lem:wFS}.
\end{remark}

The mean empirical width appears in the following important result. This theorem was originally stated in \cite{KM15} and coined \emph{Mendelson's Small Ball Method} by J. Tropp \cite{T14}. This will be a primary tool in obtaining our main estimates.

\begin{theorem}[{\cite[Proposition 5.1]{T14}}, cf. {\cite[Theorem 2.1]{KM15}}]\label{thm:QW}
	Fix a set $S \subset \mathbb{R}^n$. Let $\vf$ be a random vector in $\mathbb{R}^n$ and let $\vF \in \R^{m \times n}$ have rows $\{\vf_i^\top\}_{i = 1}^M$ that are independent copies of $\vf^\top$. Define
	\[
	Q_\xi (S; \, \vf) := \inf_{\vx \in S} \Pr\bigg(|\langle \vx, \vf \rangle| \geq \xi \bigg) \quad \mbox{and}\quad W_m(S; \, \vf) := \mathbb{E} \sup_{\vx \in S} \bigg\langle \vx, \dfrac{1}{\sqrt{m}} \sum_{i = 1}^m \e_i \vf_i \bigg\rangle,
	\]
	where $\{\e_i\}_{i = 1}^m$ are independent random variables taking values uniformly over $\{\pm 1\}$ and are independent from everything else. Then for any $\xi > 0$ and $t > 0$, we have
	\begin{align}\label{equ:msbineq}
	\inf_{\vx \in S} \| \vF \vx \|_2 \geq \xi \sqrt{m} Q_{2\xi}(S; \, \vf) - 2 W_m(S;\, \vf) - \xi t
	\end{align}
	with probability $\geq 1 - e^{-t^2/2}$.
\end{theorem}

\section{Main Results}\label{sec:main}

\subsection{Preservability Under Subgaussian Maps}

\begin{theorem}[Preservability under random maps]\label{thm:S} 
	Let \(\vD \in \R^{d \times n} \) be arbitrary, let $S \subset \R^n$ be such that $\inf\left\{\|\vD\vx\|_2 : \vx\in S\right\}\geq\eta$ for some constant \(\eta > 0 \), and assume \(\vphi \in \R^d \) is a subgaussian vector with parameters \((\alpha,\sigma) \). If $\vPhi \in \R^{m \times d}$ is a measurement matrix with rows that are independent copies of $\vphi^\top$ and that the number of measurements satisfies 
	\[
	m \geq \frac{4^8}{\eta^2}\frac{\sigma^6}{\alpha^6}C^2w^2(\vD S),
	\]
	then with probability at least  
	\[
	1 -\exp\left({- m\dfrac{\alpha^4}{64^2\sigma^4}}\right),
	\]
	we have  
	\begin{equation}\label{equ:infPhiD}
	\displaystyle\inf_{\vx\in S}\|\vPhi \vD\vx\|_2\geq C\sigma w(\vD S).
	\end{equation}
\end{theorem}

Theorem \ref{thm:S} is beyond the null space property. It says that if an operator stays bounded away from \(\vzero \) on some set, then this operator under a random map also stays bounded away from \(\vzero \) on the same set, given that the dimension of the random map is at least proportional to the square of the Gaussian width of the related set. This could be potentially useful for other dimension reduction analysis. The proof of Theorem \ref{thm:S} can be found in Section \ref{sec:mend}. In equation \eqref{equ:infPhiD}, the constant $\eta$ does not show explicitly, but it is reflected in the number of measurements $m$ since $\vPhi$ is random.

\begin{remark}
	Theorem \ref{thm:S} can be compared to \cite[Proposition 18]{CM}. Both statements are about the minimal number of measurements related to Gaussian width. However, ours has a dictionary $\vD$ incorporated.
\end{remark}

As an application of Theorem \ref{thm:S}, we let $S=S_\gamma$, as defined in \eqref{equ:S2}, and compute the Gaussian width of $\vD S_\gamma$. This is a key step in this paper and is not a simple task.  See Theorem \ref{thm:wDS} on bounding the Gaussian width. Recall that $S_\gamma$ is the set of vectors that violates the stable NSP. Theorem \ref{thm:S} and Theorem \ref{thm:wDS} imply the following theorem on preserving the null space property.
\begin{theorem}\label{thm:main}
	Assume $\vPhi \in \mathbb{R}^{m \times d}$ is a sensing matrix comprised of rows drawn i.i.d. from a subgaussian distribution with parameters $(\alpha, \sigma)$. Take $\vD \in \mathbb{R}^{d \times n}$ to be a dictionary satisfying the stable NSP of order~$s$ with $\inf\left\{\|\vD\vx\|_2 : \vx\in S_\gamma\right\}\geq\eta$ for some \(\eta> 0 \) and satisfying $\max\left\{\|\vd_i\|^2_2 : i \in [n] \right\} \leq \rho$. If  the number of measurements satisfies 
	\begin{equation}\label{equ:C}
	m\geq \frac{36\cdot 4^8}{\eta^2}\frac{\sigma^6}{\alpha^6} \dfrac{\rho}{\gamma^2}C^2 s\log(\frac{\sqrt{2}n}{s}),
	\end{equation}
	 then with probability at least  
	\[
	1 -\exp\left({- m\dfrac{\alpha^4}{64^2\sigma^4}}\right),
	\]
	the composition $\vPhi \vD$ also has the stable NSP of order $s$ with \(\inf\left\{\|\vPhi\vD\vx\|_2 : \vx\in S_\gamma\right\}\geq C \sigma w(\vD S_{\gamma})\).
\end{theorem}

\begin{remark}
	Theorem \ref{thm:main} only requires  a minimal condition on $\vD$. If $\vPhi \vD$ has the NSP, then $\vD$ must also have the NSP (hence some kind of stable NSP) since $\ker(\vD)$ is a proper subspace of $\ker(\vPhi \vD)$. Thus, $\vD$ having the NSP is a a very reasonable condition if we want stable recovery through \(\ell_1 \)-synthesis. 
\end{remark}

\begin{remark}
The constants in \eqref{equ:C} are well controlled. The ratio $\frac{\sigma}{\alpha}$ reflects how well the distributions behave and is the condition number of $\vSigma$ if $\vphi\sim N(\vzero,\vSigma)$. To reiterate, the constant $C$ that appears in the above results is a universal constant. A more precise estimate can be made when $\vphi$ follows the multivariate normal distribution. Furthermore, $\rho$ also equals one when each frame vector has unit norm, which is often the case. The constants $\eta$ and $\gamma$ reflect the null space property of $D$.
\end{remark}


\subsection{Sparse Recovery via the $\ell_1$-synthesis Method}

As a consequence of Theorem \ref{thm:main}, we obtain a uniform recovery result using Theorem \ref{thm:snsp} with $\vA$ replaced by $\vPhi\vD$. 
We will use the standard notation
\[
\sigma_s(\vx) := \inf\{\|\vx - \vv\|_1 : \vv \mbox{ is \(s\)-sparse} \}.
\]
to denote \emph{the \(\ell_1\)-error of best \(s\)-term approximation to a vector \(\vx\)}. This infimum is achieved by taking \(\vv := \vx_T \), where \(T\) is the index set containing the indices where the \(s\)-largest absolute value entries of \(\vx\) occur.

\begin{corollary}\label{cor:main}
	Suppose $\vPhi \in \mathbb{R}^{m \times d}$ is a sensing matrix comprised of rows drawn i.i.d. from a subgaussian distribution with parameters $(\alpha, \sigma)$ and suppose $\vD \in \mathbb{R}^{d \times n}$ has the stable NSP of order~$s$ with $\inf\left\{\|\vD\vx\|_2 : \vx\in S_\gamma\right\}\geq\eta$ for some \(\eta> 0 \) and satisfies $\max\left\{\|\vd_i\|^2_2 : i \in [n] \right\} \leq \rho$. Let $\vz_0=\vD\vx_0$, and the measurements \(\vy \) satisfying $\|\vy-\vPhi\vz_0\|_2\leq\epsilon$. If the number of measurements satisfies
	\[
	m\geq \frac{36\cdot 4^8}{\eta^2}\frac{\sigma^6}{\alpha^6} \dfrac{\rho}{\gamma^2}C^2 s\log(\frac{\sqrt{2}n}{s}),
	\]
	then with probability at least  
	$1 -\exp\left({- m\dfrac{\alpha^4}{64^2\sigma^4}}\right)$, the $\ell_1$-synthesis method \eqref{equ:l1e} provides a stable recovery for both the coefficients $\vx_0$ and the signal $\vz_0$ as
	\begin{align*}&\|\widehat{\vx}-\vx_0\|_2\leq \frac{2\gamma+2}{1-\gamma}\sigma_s(\vx_0)+\frac{2\epsilon}{C\sigma\eta}\\
	&\|\widehat{\vz}-\vz_0\|_2\leq\|\vD\|_2 \frac{2\gamma+2}{1-\gamma}\sigma_s(\vx_0)+\frac{2\epsilon\|\vD\|_2}{C\sigma\eta}.
	\end{align*}
\end{corollary}


\section{Proofs of the Results}\label{sec:proofs}

We will apply Mendelson's Small Ball Method  in our setting. Most of the work will be in properly estimating two important quantities, which will further lead to the need to estimate the Gaussian width \(w(\vD S_\gamma) \). This is what makes up Sections \ref{sec:mend} and \ref{sec:gw}. In doing so, we will obtain the bound on the number of measurements in Theorem \ref{thm:main} and its corollaries that forces \(\vPhi \vD \) to have the NSP with high probability given that \(\vPhi \) is a matrix made up of independent copies of a subgaussian vector and \(\vD \) satisfies the NSP. In Section \ref{sec:gw2} we prove Lemma \ref{lem:wFS} and use it to obtain a different estimate for the Gaussian width than in Section \ref{sec:gw}. Lastly, in Section \ref{sec:stab_nsp}, we provide a stable recovery result with stable NSP for completeness.

\subsection{Preservability Under a Random Map}\label{sec:mend}

Notice that if the rows of $\vPhi$ are independent copies of a random vector $\vphi$, then the rows of $\vPhi\vD$ are independent copies of the random vector \(\vD^\top \vphi \).
We will apply Theorem \ref{thm:QW} with \(\vPhi\vD\) in place of \(\vF \) and the random vector \(\vD^\top\vphi \) in place of \(\vf \), which in turn will require us to estimate the  quantities \(Q_{2\xi}(S; \, \vD^\top\vphi) \) and  \(W_m(S;\, \vD^\top\vphi)\). 

In the proof of Theorem \ref{thm:inf}, we will use the following lemma to bound \(Q_{2\xi}(S; \, \vD^\top\vphi) \). The proof can be found in \cite[Section 6.5]{T14}.

\begin{lemma}\label{lem:Q}
	If $\vf \in \R^n$ is a subgaussian vector with parameters $(\alpha, \sigma)$, then 
	\[
	\Pr\left[|\langle \vx, \vf \rangle|\geq t\right] \geq \frac{(\alpha-t)^2}{4\sigma^2}
	\]	
	for any $0<t<\alpha$ and \(\vx \in \Sb_2^{n-1}\).
\end{lemma}



\begin{theorem}\label{thm:inf}
	Let \(\vD \in \R^{d \times n} \) be arbitrary, let $S\subset \R^n $ be so that  $\inf\left\{\|\vD\vx\|_2 : \vx\in S\right\}\geq\eta$ for some constant \(\eta > 0 \), and let $\vphi \in \R^d $ be a subgaussian measurement with parameter $(\alpha, \sigma)$. If $\vPhi \in \R^{m \times d}$ has rows that are independent copies of $\vphi^\top$, then 
	\begin{equation}\label{equ:inf}
	\inf_{x \in S} \| \vPhi \vD \vx \|_2 \geq \frac{\alpha\eta}{4^3}\left(\frac{\alpha}{\sigma}\right)^2\sqrt{m} - 2C\sigma w(\vD S) - \frac{\alpha\eta}{4} t
	\end{equation}
	for any \(t > 0\) with probability at least $1 - e^{-t^2/2}$. 
\end{theorem}

\begin{proof}
	We first apply Mendelson's Small Ball Method, Theorem \ref{thm:QW}, with \(\vF \) replaced by \(\vPhi\vD\) and therefore \(\vf \) replaced by \(\vD^\top\vphi \) to obtain the bound
	\begin{equation}\label{equ:QW}
	\inf_{\vx \in S} \| \vPhi \vD \vx \|_2 \geq \xi \sqrt{m} Q_{2\xi}(S; \, \vD^\top\vphi) - 2 W_m(S;\, \vD^\top\vphi) - \xi t.
	\end{equation}
	By Lemma \ref{lem:Q}, provided we choose \(\xi \) to satisfy \(2\xi/\eta < \alpha \), we obtain for any \(\vx \in S \) 
	\begin{align}\label{equ:prdx}
	\Pr\left(|\langle \vD\vx, \vphi \rangle| \geq 2\xi \right) 
	= \Pr\left(\left|\left\langle\dfrac{\vD\vx}{\|\vD\vx\|_2}, \vphi\right\rangle\right|\geq\frac{\xi}{\|\vD\vx\|_2}\right) 
	\geq  \Pr\left(\left|\left\langle\dfrac{\vD\vx}{\|\vD\vx\|_2}, \vphi\right\rangle\right|\geq\frac{\xi}{\eta}\right)
	\geq \frac{(\alpha-2\xi/\eta)^2}{4\sigma^2}
	\end{align}
	and therefore
	\begin{align*}
	Q_{2\xi}(S; \, \vD^\top\vphi) 
	= \inf_{\vx\in S} \Pr\left(|\langle \vx, \vD^\top\vphi \rangle| \geq 2\xi \right)
	= \inf_{x\in S}\Pr\left(|\langle \vD\vx, \vphi \rangle| \geq 2\xi \right) 
	\geq \frac{(\alpha-2\xi/\eta)^2}{4\sigma^2}.
	\end{align*}
	Lemma \ref{lem:W} readily gives the estimate
	\begin{align*}
	W_m(S;\, \vD^\top\vphi) =  W_m(\vD S;\, \vphi) \leq C \sigma w(\vD S).
	\end{align*}
	Placing these two bounds into \eqref{equ:QW} and choosing \(\xi\) to satisfy \( 2\xi/\eta = \alpha/2\) gives the bound in \eqref{equ:inf}.
\end{proof}

\begin{remark}
	Notice that Theorem \ref{thm:inf} is a generalization of \cite[Theorem 6.3]{T14}. In this generalization, it is crucial that $D$ has NSP as is evident in \eqref{equ:prdx}.
\end{remark}

Finally, we can provide the proof of our first main result, Theorem \ref{thm:S}.

\begin{proof}[Proof of Theorem \ref{thm:S}]
	Theorem \ref{thm:inf} implies that 
	\begin{equation*}
	\inf_{x \in S} \| \vPhi \vD \vx \|_2 \geq \frac{\alpha\eta}{4^3}\left(\frac{\alpha}{\sigma}\right)^2\sqrt{m} - 2C\sigma w(\vD S) - \frac{\alpha\eta}{4} t := a - b - \dfrac{\alpha\eta}{4} t.
	\end{equation*}
	Picking \(m \) and \(t\) to satisfy \(a \geq 2b\) and \((\alpha \eta /4)t = (a - b)/2 \) gives
	\begin{equation*}
	\inf_{x \in S} \| \vPhi \vD \vx \|_2 \geq
	a - b - (a -b)/2 = (a - b)/2 \geq b/2 =C\sigma w(\vD S).
	\end{equation*}
	All that is left is to rewrite these  conditions in terms of \(m\) and \(t\). We have 
	\[
	a \geq 2b 
	\quad \Leftrightarrow \quad
	\frac{\alpha\eta}{4^3}\left(\frac{\alpha}{\sigma}\right)^2\sqrt{m}\geq 4C \sigma w(\vD S) 
	\quad \Leftrightarrow \quad 
	m \geq \frac{4^8}{\eta^2}\frac{\sigma^6}{\alpha^6}C^2w^2(\vD S)
	\]
	and
	\[
	\frac{\alpha\eta}{4}t=\frac{a-b}{2}\geq\frac{a}{4}=\frac{\alpha\eta}{4^4}\left(\frac{\alpha}{\sigma}\right)^2\sqrt{m} \quad \Leftrightarrow \quad  t\geq \frac{1}{64}\sqrt{m}\left(\frac{\alpha}{\sigma}\right)^2 \quad \Leftrightarrow \quad -\dfrac{t^2}{2} \leq - m\dfrac{\alpha^4}{64^2\sigma^4},
	\]
	proving the result.
\end{proof}


\subsection{Estimating the Gaussian Width}\label{sec:gw}
In order to prove Theorem \ref{thm:main}, we apply Theorem \ref{thm:S} to the null space property, i.e., let $S=S_\gamma$.
Therefore the last ingredient of the proof of Theorem \ref{thm:main} is to suitably bound the Gaussian width \(w(\vD S_\gamma)\). Note that it is pivotal to have a relatively optimal upper bound on $w(DS_\gamma)$ since the number of measurements $m$ is on the order of the square of this width. 


Recall that
\begin{equation}\label{equ:wDS}
w(\vD S_\gamma) := \mathbb{E} \sup_{\vz \in \vD S_\gamma} \langle \vg , \vz \rangle = \mathbb{E} \sup_{\vx \in S_\gamma} \langle \vD^\top \vg , \vx \rangle,
\end{equation}
where \(\vg \in \R^d \) is a standard Gaussian vector. 

We again point out that bounding  $w(\vD S_{\gamma})$ is not an easy task. At first glance, one may naively estimate \(w(\vD S_\gamma)\) using its geometric properties. We have
$$w(\vD S_\gamma ) = w(\conv(\vD S_\gamma )) = w(\vD \conv(S_\gamma )) \leq w(\vD \bB_2^{n-1}),$$ where \(\conv(S) \) denotes the convex hull of \(S\). By Lemma \ref{lem:wFS}, we obtain the crude estimate
$$w(\vD \bB_2^{n-1} ) \leq2\|\vD\|_2 w(\bB_2^{n-1})\approx 2\|D\|_2\sqrt{n}.$$
This is far less than ideal with $\sqrt{n}$ and  the potential dimension dependency from $\|D\|_2$.

The approach we will use is inspired by \cite[Section 9.4]{FRbook}, which estimates $w(S_\gamma)$, the width in the basis case. It is worth noting that the generalization  $w(\vD S_\gamma)$ is nontrivial, as will be explained in the proof of Theorem \ref{thm:wDS}. 

First, we introduce the convex cone
\begin{align*}
K_{\gamma,s} := \left\{\vu \in \R^n : \vu_\ell \geq 0, \, \ell \in [n], \, \sum_{\ell=1}^s \vu_\ell \geq \gamma \sum_{ \ell=s + 1}^n \vu_\ell \right\}
\end{align*}
and its dual cone
\begin{align*}
K_{\gamma,s}^* := \left\{\vz \in \R^n : \langle \vz,\vu\rangle \geq 0 \mbox{ for all } \vu \in K_{\gamma,s}  \right\}.
\end{align*}
Also, recall the \emph{nonincreasing rearrangement} of a vector $\vx\in\R^d$ is the vector $\vx^*\in\R^d$ for which
$$x_1^*\geq x_2^*\geq\cdots\geq x_d^*\geq0$$
and $x^*_i=\left|x_{\pi(i)}\right|$ for some permutation $\pi$.

We can now state the following lemma. It has a proof that is similar to that of  \cite[Proposition 9.31]{FRbook}, but we provide it for the sake of completeness.
\begin{lemma}\label{lem:wDS}
	If \(\vD \in \R^{d \times n} \) is an arbitrary matrix and \(\vg \in \R^d \sim N(\vzero,\vI_d) \), then
	\begin{align*}
	w(\vD S_\gamma) \leq \E \min_{\vz \in K_{\gamma,s}^*  }\|(\vD^\top \vg)^* + \vz\|_2.
	\end{align*} 
\end{lemma}

\begin{proof} 
	Elements in \(S_\gamma\) are invariant under permutation of indices  and entrywise sign changes, so 
	\begin{align}\label{equ:du}
	\max\limits_{\vx \in S_\gamma} \langle \vD^\top \vg, \vx \rangle 
	&= \max\limits_{\vx \in S_\gamma} \langle (\vD^\top \vg)^*, \vx \rangle 
	= \max\limits_{\vu \in K_{\gamma,s} \cap \mathbb{S}^{n-1}} \langle (\vD^\top \vg)^*, \vu \rangle
	\leq \min_{\vz \in K^*_{\gamma,s}  }\|(\vD^\top \vg)^* + \vz\|_2,
	\end{align}
	where the last inequality follows by the duality 
	\begin{equation*}
	\max_{\vx\in K,\|\vx\|_2\leq 1}\langle -\vg,\vx\rangle\leq\min_{\vz\in K^*}\|\vg-\vz\|_2,
	\end{equation*}
	as given in \cite[(B.39)]{FRbook}. Taking the expected value of \eqref{equ:du} completes the proof.
\end{proof}

The following Lemma is similar to \cite[Remark 9.25]{FRbook}, but here we assume a general variance.
\begin{lemma}\label{lem:soft}
	If $a\sim N(0,\sigma^2)$, then
	$$\E S_t^2(a)\leq \sigma^4 \sqrt{\dfrac{2}{\pi e}} t^{-2}  e^{-t^2/(2\sigma^2)},$$
	where \(S_t\) is the soft thresholding operator defined by
	\[
	S_t(u) := \left\{\begin{array}{l@{\quad \mbox{if }}l} 
	u - t & u > t\\
	0 & |u| \leq t\\
	u + t & u < -t
	\end{array} \right. .
	\] 
\end{lemma}
\begin{proof} We compute
	\begin{align*}
	\mathbb{E}S_t^2(a) &= \dfrac{2}{\sqrt{2\pi \sigma^2}} \int_0^\infty S_t^2(u) e^{-u^2/(2\sigma^2)} \, du \\
	&= \dfrac{1}{\sigma} \sqrt{\dfrac{2}{\pi}} \int_t^\infty (u-t)^2 e^{-u^2/(2\sigma^2)} \, du\\
	&= \dfrac{1}{\sigma} \sqrt{\dfrac{2}{\pi}} \int_0^\infty v^2 e^{-(v+t)^2/(2\sigma^2)} \, dv \\
	&= e^{-t^2/(2\sigma^2)} \dfrac{1}{\sigma} \sqrt{\dfrac{2}{\pi}} \int_0^\infty v^2 e^{-v^2/(2\sigma^2)} e^{-vt/\sigma^2} \, dv\\
	&= e^{-t^2/(2\sigma^2)}  \sqrt{\dfrac{2}{\pi}} \int_0^\infty \left( \dfrac{v}{\sigma} e^{-(v/\sigma)^2/2} \right) v  e^{-vt/\sigma^2} \, dv\\
	&\leq e^{-t^2/(2\sigma^2)}  \sqrt{\dfrac{2}{\pi}} e^{-1/2} \int_0^\infty v  e^{-vt/\sigma^2} \, dv\\
	&= e^{-t^2/(2\sigma^2)}  \sqrt{\dfrac{2}{\pi}}e^{-1/2} \frac{\sigma^4}{t^2}\\
	&= \sigma^4 \sqrt{\dfrac{2}{\pi e}} t^{-2}  e^{-t^2/(2\sigma^2)}
	\end{align*}
	as desired.
\end{proof}


The following Lemma is a key step in suitably bounding the Gaussian width. It is inspired by \cite[Lemma 3.5]{MRW18}.

\begin{lemma}\label{lem:key}
	If \(\vD \in \R^{d \times n}\) is any matrix with $\max\left\{\|\vd_i\|^2_2 : i \in [n] \right\} \leq \rho$ and \(\vg \in \R^d \sim N(\vzero,\vI_d) \), then
	$$
	\mathbb{E} \sqrt{\frac{1}{s}\sum_{\ell=1}^s ((\vD^\top \vg)^*_\ell)^2} \leq\sqrt{4\rho\log(\sqrt{2}n/s)}.
	$$
\end{lemma}
\begin{proof}
	Setting $a_i := (\vD^\top \vg)_i$ gives $a_i \sim N(0, \sigma_i^2)$ with \(\sigma_i = \|\vd_i\|_2\), so
	\begin{align*}
	\E\exp(\frac{a_i^2}{4\sigma_i^2})
	&=\int_{\R}\frac{1}{\sqrt{2\pi\sigma_i^2}}e^{-\frac{x^2}{2\sigma_i^2}}e^{\frac{x^2}{4\sigma_i^2}}dx
	=\int_{\R}\frac{1}{\sqrt{2\pi\sigma_i^2}}e^{-\frac{x^2}{4\sigma_i^2}}dx=\sqrt{2}.
		\end{align*}
	It follows from Jensen's inequality that
	\begin{align*}
	&\mathbb{E} ^2\sqrt{\frac{1}{s} \sum_{\ell=1}^s((\vD^\top \vg)^*_\ell)^2}\leq\mathbb{E}\frac{1}{s}\sum_{\ell=1}^s(a^*_\ell)^2=\mathbb{E}\frac{1}{s}\sum_{\ell=1}^s4(\sigma_l^*)^2\log\left(\exp\frac{(a_l^*)^2}{4(\sigma_l^*)^2}\right)\\
	\leq&4\rho \mathbb{E}\frac{1}{s}\sum_{\ell=1}^s\log\left(\exp\frac{(a_l^*)^2}{4(\sigma_l^*)^2}\right)\leq4\rho\mathbb{E}\log\left(\frac{1}{s}\sum_{l=1}^s\exp\frac{(a_l^*)^2}{4(\sigma_l^*)^2}\right)\\
	\leq&4\rho\log\left(\frac{1}{s}\sum_{l=1}^s\mathbb{E}\exp\frac{(a_l^*)^2}{4(\sigma_l^*)^2}\right)\leq 4\rho\log\left(\frac{1}{s}\sum_{l=1}^n\mathbb{E}\exp\frac{a_l^2}{4\sigma_l^2}\right)=4\rho\log(\frac{1}{s}\sqrt{2}n).
	\end{align*}
	Taking the squared root of this inequality yields the desired result.\end{proof}

\begin{theorem}\label{thm:wDS}
	If $\vD \in \R^{d \times n}$ is any matrix with $\max\left\{\|\vd_i\|^2_2 : i \in [n] \right\} \leq \rho$, then
	$$
	w(\vD S_\gamma)\leq 6\gamma^{-1}\sqrt{s\rho\log(\sqrt{2}n/s)}.
	$$
\end{theorem}
\begin{proof}
	Define
	\begin{align*}
	F_{\gamma,s} := \bigcup_{t \geq 0} \left\{\vz \in \R^n : z_\ell = t, \, \ell \in [s], \, z_k \geq -\gamma t,\, k = s+1,\ldots,n \right\}
	\end{align*}
	which satisfies \(F_{\gamma,s} \subset K_{\gamma,s}^*\) by Lemma 9.32 of \cite{FRbook}. The minimum over a smaller set can only be larger so we obtain \(w(\vD S_\gamma) \leq \E\min\limits_{z \in F_{\gamma,s}  }\|(\vD^\top \vg )^* + \vz\|_2\) by Lemma \ref{lem:wDS}. 
	By the definition of \(F_{\gamma, s}\) and concavity of the square root function, we get
	\begin{align*}
	w(\vD S_\gamma)
	&\le \E\min\limits_{\vz \in F_{\gamma,s}  }\|(\vD^\top \vg)^* + \vz\|_2\\
	&\leq \E \min\limits_{\underset{z_\ell \geq -\gamma t, \ell \in [n]\backslash[s]}{t\geq 0}}  \left(\sqrt{\sum_{\ell=1}^s ((\vD^\top \vg)^*_\ell + t)^2} +  \sqrt{\sum_{\ell=s+1}^n ((\vD^\top \vg)^*_\ell + z_\ell)^2} \right)
	\end{align*}
	Now fix \(t \ge 0\) to be chosen later. Again by concavity, and since \((\vD^\top \vg)^*_\ell\) has mean zero we obtain
	\begin{align}
	w(\vD S_\gamma) 
	&\leq \E \sqrt{\sum_{\ell=1}^s ((\vD^\top \vg)^*_\ell + t)^2} + \E   \sqrt{\min\limits_{z_\ell \geq -t} \sum_{\ell=s+1}^n ((\vD^\top \vg)^*_\ell + z_\ell)^2} \\
	&\leq \E \sqrt{\sum_{\ell=1}^s ((\vD^\top \vg)^*_\ell)^2} + \sqrt{2t\, \sum_{\ell=1}^s\E (\vD^\top \vg)_\ell^*} + t\sqrt{s} + \E   \sqrt{\min\limits_{z_\ell \geq -t} \sum_{\ell=s+1}^n ((\vD^\top \vg)^*_\ell + z_\ell)^2} \\\label{equ:soft}
	&\leq \mathbb{E} \sqrt{\sum_{\ell=1}^s ((\vD^\top \vg)^*_\ell)^2} + t \sqrt{s} + \E\sqrt{   \sum_{\ell = s+1}^n S_{\gamma t}^2 ((\vD^\top \vg)^*_\ell)},
	\end{align}
	where \(S_t\) is the soft thresholding operator defined earlier.
	Note that Inequality \eqref{equ:soft} holds because of how the $z_l$ are defined.
	
We have so far exactly followed the proof of \cite[Proposition 9.33]{FRbook}, but if one tries to estimate the first term of \eqref{equ:soft} in the same fashion as in \cite[Proposition 8.2]{FRbook}, a $\sqrt{d}$ factor would result, causing the estimate to be too large. Instead, we apply Lemma \ref{lem:key}, and get
	\[
	\mathbb{E} \sqrt{\sum_{\ell=1}^s ((\vD^\top \vg)^*_\ell)^2}
	\leq \sqrt{4\rho s\log(\sqrt{2}n/s)}.
	\]	
	Next, we bound the last term in \eqref{equ:soft} by
	$$
	\E\sqrt{   \sum_{\ell = s+1}^n S_{\gamma t}^2 ((\vD^\top \vg)^*_\ell)}
	\leq \sqrt{\E\sum_{\ell = s+1}^n S_{\gamma t}^2 ((\vD^\top \vg)^*_\ell)}
	\leq \sqrt{\E\sum_{\ell = s+1}^n S_{\gamma t}^2 ((\vD^\top \vg)_\ell)},
	$$
	where we used the fact that $\{(\vD^\top \vg)^*_\ell : \ell=s+1,\cdots,n \}$ are the  $n-s$ smallest entries in magnitude to obtain the second inequality. 
	
	To estimate the second moment of the soft thresholding operator, we again use the fact that $a_i := (\vD^\top \vg)_i$ gives $a_i \sim N(0, \sigma_i^2)$ with \(\sigma_i = \|\vd_i\|_2\) and so Lemma \ref{lem:soft} implies
	\begin{align*}
	\mathbb{E}S_{\gamma t}^2(a_i) \leq \sigma_i^4 \sqrt{\dfrac{2}{\pi e}} (\gamma t)^{-2}  \exp\left(\dfrac{-(\gamma t)^2}{2\sigma_i^2}\right)\leq \rho^2 \sqrt{\dfrac{2}{\pi e}} (\gamma t)^{-2} \exp\left(\dfrac{-(\gamma t)^2}{2\rho}\right)
	\end{align*}
	Finally, combining all of these estimates of the quantities in \eqref{equ:soft} gives
	\begin{align*}
	w(\vD S_\gamma) &\le  \sqrt{4\rho s\log(\sqrt{2}n/s)}
	+ t\sqrt{s} 
	+ \sqrt{ (n - s) \rho^2 \sqrt{\dfrac{2}{\pi e}} (\gamma t)^{-2} \exp\left(\dfrac{-(\gamma t)^2}{2\rho}\right)}.
	\end{align*}
	Choosing \(t = \gamma^{-1}\sqrt{4\rho \log(\sqrt{2}n/s)} \) and using the fact that \(\gamma^{-1} \geq 1 \) gives
	\begin{align*}
	w(\vD S_\gamma) &\le \sqrt{4\rho s\log(\sqrt{2}n/s)} + \gamma^{-1}\sqrt{4s\rho \log(\sqrt{2}n/s)} + \sqrt{\rho\dfrac{(n-s)s^2}{2n^2}  \sqrt{\dfrac{2}{\pi e  }}\dfrac{1}{\log(\sqrt{2}n/s)}}\\
	&\le \sqrt{4\rho s\log(\sqrt{2}n/s)}\left(1 + \gamma^{-1} \right) +\sqrt{\rho s}\\
	&\le 3\gamma^{-1}\sqrt{4\rho s\log(\sqrt{2}n/s)}.
	\end{align*}
	\end{proof}

%


\subsection{Proof of Lemma \ref{lem:wFS}}\label{sec:gw2}


Recall that a \emph{Gaussian process} \(\{\vX_t\}_{t\in T} \) for some index set \(T \) (which can be uncountably infinite) is a sequence of random variables \(\vX_t \) so that any finite linear combination follows a Gaussian distribution. Slepian's lemma \cite{LT91} gives a way to compare such processes. 
\begin{lemma}[Slepian's Lemma]\label{lem:slep}
	If \(\{\vX_t \}_{t \in T} \) and \(\{\vY_t \}_{t \in T} \) are Gaussian processes so that for any \(s,t \in T, \)
	\[
	\bE |\vX_s - \vX_t|^2 \le \bE |\vY_s - \vY_t|^2
	\]
	holds, then 
	\[
	\bE \sup_{t \in T} \vX_t \leq \bE \sup_{t \in T} \vY_t.
	\]
\end{lemma}

\begin{proof}[Proof of Lemma \ref{lem:wFS}]
We define the Gaussian processes \(\{\vX_\vu \}_{\vu \in S} \) and \(\{\vY_\vu \}_{\vu \in S} \)
	\begin{align*}
	\vX_\vu &:=  \langle \vF\vu, \vg\rangle = \|\vF\vu \|_2 \left \langle \dfrac{ \vF\vu }{\| \vF\vu \|_2}, \vg \right\rangle  
	\quad \mbox{and} \quad 
	\vY_\vu := \|\vF\|_2 \langle \vu, \vg\rangle,
	\end{align*}
	where $g\sim N(\vzero, \vI_d)$. We notice 
	\[
	\E \left|\vX_\vu - \vX_\vv\right|^2 
	= \left\|\vF(\vu - \vv)\right\|_2^2 	
	\leq \|\vF\|_2^2 \|\vu - \vv\|_2^2
	= \E \left|\vY_\vu - \vY_\vv\right|^2.
	\]
	Thus, Lemma \ref{lem:slep} gives
	\[
	w(\vF S) = \E\sup_{\vu \in S} \vX_\vu \leq \E\sup_{\vu \in S} \vY_\vu = \|\vF\|_2 w(S). 
	\]	
\end{proof}

\subsection{Sparse recovery with Stable NSP}\label{sec:stab_nsp}

\begin{proof}[Proof of Theorem \ref{thm:snsp}]
Let $\vh:=\widehat{\vx}-\vx$ and $T$ be the support of the biggest $s$ entries of $\vx$ in magnitude. Then by a standard compressed sensing argument, we have
\begin{equation}\label{equ:htc}
\|\vh_{T^C}\|_1\leq\|\vh_T\|_1+2\sigma_s(\vx),
\end{equation}
where $\sigma_s(\vx)=\|\vx-\vx_T\|_1$.

If $\vh/\|\vh\|_2\in S_\gamma$, then 
\begin{equation}\label{equ:e}
\eta\|\vh\|_2\leq\|\vA\vh\|_2\leq2\epsilon.
\end{equation}
On the other hand, if $\vh/\|\vh\|_2\notin S_\gamma$, then the vector $\vh$ itself has the stable NSP and therefore $\|\vh_T\|_1<\gamma\|\vh_{T^C}\|_1$. Combined with \eqref{equ:htc}, we have 
\begin{equation}\label{equ:ht}
\|\vh_{T}\|_1\leq\frac{2\gamma}{1-\gamma}\sigma_s(\vx).
\end{equation}
The equations \eqref{equ:htc}, \eqref{equ:ht}, and the fact that $\|\vh\|_1=\|\vh_T\|_1+\|\vh_{T^C}\|_1$, we get
\begin{equation}\label{equ:h}
\|\vh\|_2\leq\|\vh\|_1\leq\frac{2\gamma+2}{1-\gamma}\sigma_s(\vx).
\end{equation}
Combining the two cases \eqref{equ:e} and \eqref{equ:h}, we get the desired result \eqref{equ:snspresult}.
\end{proof}

\section{ Estimates For Gaussian Distributions}\label{sec:gauss}

The Gaussian distributions are important special cases of subgaussian distributions, so we list this special cases of Theorem \ref{thm:main} below, using the estimates in Example \ref{ex:non}, as well as Remark \ref{rem:C}. In Corollary \ref{cor:non}, $\kappa:=\sigma_{\max}^2/\sigma_{\min}^2$ is the condition number of $\vSigma$.
\begin{corollary}\label{cor:non}
	Suppose $\vPhi  \in \mathbb{R}^{m \times d} $ is a sensing matrix with rows drawn i.i.d. from a Gaussian distribution $N(\vzero,\vSigma)$ in which the covariance matrix $\vSigma$ has condition number $\kappa$, and suppose $\vD \in \mathbb{R}^{d \times n}$ has the stable NSP of order~$s$ with $\inf\left\{\|\vD\vx\|_2 : \vx\in S_\gamma\right\}\geq\eta$ for some \(\eta> 0 \) and satisfies $\max\left\{\|\vd_i\|^2_2 : i \in [n] \right\} \leq \rho$. If the the number of measurements satisfies
	\[
	m\geq \frac{9\cdot2^{15}\pi^3}{\eta^2}\dfrac{\rho\kappa^3}{\gamma^2} s\log(\frac{\sqrt{2}n}{s}),
	\]
	then with probability at least  
	\[
	1 -\exp\left(- m\frac{\kappa^2}{4^5 \pi^2 }\right),
	\]
	the composition $\vPhi \vD$ also has the stable NSP of order $s$ with \(\inf\left\{\|\vPhi\vD\vx\|_2 : \vx\in S_\gamma\right\}\geq \sigma_{\max} w(\vD S_{\gamma})\).
\end{corollary}


The following theorem is an improvement on Corollary \ref{cor:non} in terms of the condition number $\kappa$. The techniques are the same ones used in proving Theorem \ref{thm:main}, but for the Gaussian distribution we can improve Lemma \ref{lem:Q} and hence improve the estimate on the marginal tail $Q_{\xi}$. See \eqref{equ:Q} below.

\begin{theorem}\label{thm:main_gauss}
	Suppose $\vPhi  \in \mathbb{R}^{m \times d} $ is a sensing matrix with rows drawn i.i.d. from a Gaussian distribution $N(\vzero,\vSigma)$ in which the covariance matrix $\vSigma$ has condition number $\kappa$, and suppose $\vD \in \mathbb{R}^{d \times n}$ has the NSP of order~$s$ with $\inf\left\{\|\vD\vx\|_2 : \vx\in S_\gamma\right\}\geq\eta$ for some \(\eta> 0 \) and satisfies $\max\left\{\|\vd_i\|^2_2 : i \in [n] \right\} \leq \rho$. If the the number of measurements satisfies
	\[
	m \geq \dfrac{18 \cdot 2^9 \pi e}{\eta^2} \dfrac{\rho\kappa}{\gamma^2} s\log(2n),
	\]
	then with probability at least  
	\[
	1 -\exp\left(- m\frac{1}{128 e \pi}\right),
	\]
	the composition $\vPhi \vD$ also has the NSP of order $s$. 
\end{theorem}


\begin{proof}
	Let \(\vphi \in \R^d \sim N(\vzero,\vSigma) \) be the vector that generates \(\vPhi \). 
	The random vector \(\vphi \) has the same distribution as \(\vSigma^{1/2} \vg \), where \(\vg \sim N(\vzero, \vI_d) \). Therefore,
	\[
	\langle \vD\vx,\vphi\rangle=\langle \vD\vx, \vSigma^{1/2} \vg \rangle
	= \langle \vSigma^{1/2} \vD\vx,\vg \rangle 
	= \|\vSigma^{1/2} \vD\vx \|_2 \bigg \langle \dfrac{\vSigma^{1/2} \vD\vx }{\|\vSigma^{1/2} \vD\vx \|_2}, \vg \bigg\rangle 
	= \|\vSigma^{1/2} \vD\vx \|_2 \, Z
	\]
	where \(Z \sim N(0,1) \). Therefore, we obtain
	\begin{align}\notag
	Q_{2\xi}(S_\gamma; \, \vD^\top\vphi) 
	&=\inf_{\vx \in S_\gamma }\Pr\left(\left|\left\langle \vD\vx, \vphi \right\rangle\right| \geq 2\xi \right) \\\notag
	&= \inf_{\vx \in S_\gamma }\Pr\left(|Z| \geq \dfrac{2\xi}{\left\|\vSigma^{1/2}\vD\vx \right\|_2} \right) \\\notag
	&\geq \Pr\left(|Z| \geq \dfrac{2\xi}{\eta \sigma_{\min}} \right) \\\label{equ:Q}
	&\geq \dfrac{\eta \sigma_{\min}}{4\xi }\cdot \dfrac{1}{\sqrt{2\pi}} \exp\left(-\dfrac{2\xi^2}{\eta^2\sigma_{\min}^2}\right)
	\end{align}
	where we used the well-known lower bound 
	\begin{align*}
	\Pr(Z>t)\geq \frac{1}{2t}\frac{1}{\sqrt{2\pi}}e^{-t^2/2}, \qquad\text{ for any }t\geq1
	\end{align*}
	with the assumption that $2 \xi/(\eta \sigma_{\min}) \geq 1$. We also have	
	\[
	W_m(S_\gamma;\, \vD^\top\vphi) = W_m(\vD S_\gamma;\, \vD^\top\vphi) \leq \sigma_{\max} w(\vD S_\gamma)
	\]	
	by Remark \ref{rem:C}.
	Theorem \ref{thm:QW} hence implies that
	\begin{align*}
	\inf_{\vx \in S_\gamma} \| \vPhi \vD \vx\|_2 
	&\geq\xi\sqrt{m} Q_{2\xi}(S_\gamma; \, \vD^\top\vphi) - 2 W_m(S_\gamma;\, \vD^\top\vphi) - \xi t \\
	&\geq \sqrt{m}  \dfrac{\eta \sigma_{\min}}{4\sqrt{2\pi}} \exp\left(-\dfrac{2\xi^2}{\eta^2\sigma_{\min}^2}\right) - 2\sigma_{\max} w(\vD S_\gamma) - \xi t\\ &=: a - b - \xi t
	\end{align*}
	with probability $\geq 1 - e^{-t^2/2}$.	Picking $m, \xi, t$ such that $a\geq 2b$ and $\xi t=(a-b)/2$ gives
	$$
	\inf_{\vx \in S} \| \vPhi \vD \vx\|_2\geq a-b-(a-b)/2=(a-b)/2\geq b/2 > 0,
	$$
	guaranteeing the null space property of $\vPhi \vD$. 
	
	Lastly, we rewrite these conditions choosing \(2\xi = \eta \sigma_{\min} \) and invoke Theorem \ref{thm:wDS} to get
	\[
	a \geq 2b
	\quad \Leftrightarrow \quad 
	\sqrt{m}  \dfrac{\eta \sigma_{\min}}{4\sqrt{2\pi}} e^{-1/2} \geq 4\sigma_{\max} w(\vD S_\gamma) 
	\quad \Leftrightarrow \quad 
	m \geq \dfrac{18 \cdot 2^9 \pi e}{\eta^2} \dfrac{\rho  \kappa}{\gamma^2}  s\log(2n)
	\]
	and
	\[
	\dfrac{\eta \sigma_{\min}}{2}t = \dfrac{a-b}{2} \geq \dfrac{a}{4} = \sqrt{m}\dfrac{\eta \sigma_{\min}}{16\sqrt{2\pi}} e^{-1/2}	
	\quad \Leftrightarrow \quad 
	t \geq \dfrac{\sqrt{m}}{8\sqrt{2\pi e}}
	\quad \Leftrightarrow \quad 
	-\dfrac{t^2}{2} \leq -m \dfrac{1}{128\pi e}
	\]
	as desired.
\end{proof}

\section{Conclusions and Discussions}
This paper generalizes recovery guarantees for $\ell_1$-synthesis with dictionary sparsity, which is of interesting in the field of compressed sensing. In particular, the recovery result Corollary \ref{cor:main} states that  $O(s\log(n/s))$ subgaussian measurements are enough while only imposing minimal conditions on $\vD$, if $\vD$ is full spark. As far as the authors can tell, such work is first of its kind. In proving the result, the first step is a rather classical application of the Mendelson's small ball method. The second step, the Gaussian width computation, is the main technical contribution.


Similar recovery results for the basis case $\vD=\vI$ was generalized to random distributions far beyond subgaussian~\cite{DLR18, LM17}. It would be  interesting to consider more general random matrices such as subexponential distributions.
Another related direction of future work in the non-standard Gaussian case is to explore redefining the inner product as $x^T\Sigma^{-1/2}y$. Similar results as obtained above might hold true in this non-isotropic setting without extensive adjustments to our arguments.


\section*{Acknowledgements}


We would like to thank Eric Pinkham for taking part in our initial discussions to get the project moving.
We also wish to thank Simon Foucart for his extremely helpful insights and suggestions. We thank the anonymous referees for their extremely helpful comments which improved the presentation of the paper. 
Finally, we are very gracious for the grant support we received: ARO W911NF-16-1-0008; NSF ATD 1321779; NSF DMS 1307685; and NSF DMS 1725455.


\begin{thebibliography}{WW}
	
	
	\bibitem{ACM12}
	B.\ Alexeev, J.\ Cahill, and D.\ G.\ Mixon. \emph{Full spark frames.} J. Fourier Anal. Appl. 18 (2012), no. 6: 1167--1194.
	
	\bibitem{ACP12}
	A.\ Aldroubi,  X.\ Chen, and A.\ Powell. \emph{Perturbations of measurement matrices and dictionaries in compressed sensing.} Appl. Comput. Harmon. Anal. 33 (2012), no. 2: 282--291.
	
	\bibitem{ACP11}
	A.\ Aldroubi,  X.\ Chen, and A.\ Powell. \emph{Stability and robustness of \(\ell_q \) minimization using null space property.} Proceedings of SampTA 2011 (2011).
	
	\bibitem{BDFKK11}
	J. Bourgain, S.J. Dilworth, K. Ford, S. Konyagin, and D. Kutzarova. \emph{Explicit constructions of RIP matrices and related problems.} Duke Mathematical Journal 159.1 (2011): 145-185.
	
	\bibitem{CCW16}
	J.\ Cahill, X.\  Chen, and R.\ Wang. \emph{The gap between the null space property and the restricted isometry property.} Linear Algebra and its Applications 501 (2016): 363--375.
	
	\bibitem{CM}
	J.\ Cahill and D.\ G.\ Mixon. \emph{Robust width: A characterization of uniformly stable and robust compressed sensing.} arXiv:1408.4409
	
	\bibitem{CENR11}
	E.\ J. Cand\`{e}s, Y.\ Eldar, D.\ Needell, and P.\ Randall. \emph{Compressed sensing with coherent and redundant dictionaries.} Appl. Comput. Harmon. Anal. 31 (2011), no. 1: 59--73.
	
	\bibitem{CRT06}
	E.\ J.\ Cand\`{e}s, J.\ Romberg, and T.\ Tao. \emph{Stable signal recovery from incomplete and inaccurate measurements.} Comm. Pure Appl. Math. 59 (2006), no. 8: 1207--1223.
	
	\bibitem{CT05}
	E.\ J.\ Cand\`{e}s and T.\ Tao. \emph{Decoding by linear programming.} IEEE Trans. Inform. Theory 51 (2005), no. 12:  4203--4215.
	
	\bibitem{CT06}
	E.\ J.\ Cand\`{e}s and T.\ Tao. \emph{Near optimal signal recovery from random projections: universal encoding strategies?} IEEE Trans. Inform. Theory 52 (2006), no. 12: 5406--5425.
	
	\bibitem{CK13} 
	P.\ G.\ Casazza and G.\ Kutyniok. ``Finite frames: Theory and applications.'' Springer (2013).
	
	\bibitem{CL16}
	P.\ G.\ Casazza and R.\ G.\ Lynch. \emph{A brief introduction to Hilbert space frame theory and its applications. Finite frame theory.} Proc. Sympos. Appl. Math., 73, AMS Short Course Lecture Notes, Amer. Math. Soc., Providence, RI, (2016): 1--51.
	
	\bibitem{CDS99}
	S.\ S.\ Chen, D.\ L.\ Donoho, and M.\ A.\ Saunders. \emph{Atomic decomposition by Basis Pursuit.} SIAM J. Sci. Comput. 20 (1998), no. 1: 33-61. Reprinted in SIAM Rev. 43 (2001), no. 1: 129--159.
	
	\bibitem{CWW14}
	X.\ Chen, H.\ Wang, and R.\ Wang. \emph{A null space analysis of the $\ell_1$-synthesis method in dictionary-based compressed sensing.} Applied and Computational Harmonic Analysis 37 (2014), no. 3: 492--515.
	
	\bibitem{CRPW12}
	V.\ Chandrasekaran, B.\ Recht, P.\ Parrilo, and A.\ Willsky. \emph{The convex geometry of linear inverse problems.} Found. Comput. Math. 12 (2012), no. 6: 805--849
	
	\bibitem{CDD17}
	A.\ Cohen, W.\ Dahmen, and R.\ DeVore. \emph{Orthogonal Matching Pursuit under the Restricted Isometry Property.} Constr. Approx. 45 (2017), no. 1: 113--127.
	
	\bibitem{DNW12}
	M.\ A.\ Davenport, D.\ Needell, and M.\ B.\ Wakin. Signal space CoSaMP for sparse recovery with redundant dictionaries. IEEE Trans. Inform. Theory 59 (2013), no. 10: 6820--6829.
	
	\bibitem{DPW09}
	R. DeVore,  G. Petrova, and P. Wojtaszczyk. \emph{Instance-optimality in probability with an $\ell_1$-minimization decoder.} Applied and Computational Harmonic Analysis 27.3 (2009): 275-288.
	
	\bibitem{DLR18}
	S. Dirksen,  G. Lecu\'e, and H. Rauhut. \emph{On the gap between restricted isometry properties and sparse recovery conditions.} IEEE Transactions on Information Theory 64.8 (2016): 5478-5487.
	
	\bibitem{DE03}
	D.\ Donoho and M.\ Elad. \emph{Optimally sparse representation in general (nonorthogonal) dictionaries via $\ell_1$ minimization.} Proceedings of the National Academy of Sciences 100 (2003), no. 5: 2197--2202.
	
	\bibitem{F11}
	S. Foucart,  \emph{Hard thresholding pursuit: an algorithm for compressive sensing.} SIAM Journal on Numerical Analysis 49.6 (2011): 2543-2563.
	
	\bibitem{F14}
	S.\ Foucart. \emph{Stability and robustness of $\ell_1$-minimizations with Weibull matrices and redundant dictionaries.} Linear Algebra Appl. 441 (2014): 4--21.
	
	\bibitem{F16}
	S.\ Foucart. \emph{Dictionary-sparse recovery via thresholding-based algorithms.} J. Fourier Anal. Appl. 22 (2016), no. 1: 6--19.
	
	\bibitem{FSurvey}
	S.\ Foucart. \emph{Flavors of Compressive Sensing.} In: Fasshauer G., Schumaker L. (eds) Approximation Theory XV: San Antonio 2016. AT 2016. Springer Proceedings in Mathematics \& Statistics, vol 201. Springer, Cham.
	
	\bibitem{FL09}
	S.\ Foucart and M.\ Lai. \emph{Sparsest solutions of underdetermined linear systems via \(\ell_q \)-minimization for \(0 < q \leq 1 \).} Appl. Comput. Harmon. Anal. 26 (2009), no. 3: 395-407.
	
	
	\bibitem{FRbook}
	S.\ Foucart and H.\ Rauhut. ``A mathematical introduction to compressive sensing.'' Boston: Birkh\"{a}user (2013).
	
	\bibitem{GN03}
	R.\ Gribonval and M.\ Nielsen, \emph{Sparse representations in unions of bases.} IEEE Trans. Inform. Theory 49 (2003), no. 12: 3320--3325.
	
	\bibitem{LT91}
	M.\ Ledoux and M.\ Talagrand. ``Probability in Banach spaces: isoperimetry and processes.'' Volume 23. Springer (1991).
	
	\bibitem{LML12}
	S.\ Li, T.\ Mi, and Y.\ Liu. \emph{Performance analysis of \(\ell_1 \)-synthesis with coherent frames.} 2012 IEEE International Symposium on Information Theory Proceedings (2012).
	
	\bibitem{KR15}
	M.\ Kabanava  and H.\ Rauhut. \emph{Analysis $\ell_1$-recovery with frames and Gaussian measurements.} Acta Appl. Math. 140 (2015): 173--195.
	
	\bibitem{KKZ14}
	E. J. King,  G. Kutyniok, and X. Zhuang. \emph{Analysis of inpainting via clustered sparsity and microlocal analysis}" Journal of mathematical imaging and vision 48.2 (2014): 205-234.
	
	\bibitem{KM15}
	V.\ Koltchinskii and S.\ Mendelson. \emph{Bounding the smallest singular value of a random matrix without concentration.} International Mathematics Research Notices. 2015 (2015), no. 23: 12991--13008.
	
	\bibitem{LM17}
	G. Lecu\'e,  and S. Mendelson. \emph{Sparse recovery under weak moment assumptions.} J. Eur. Math. Soc., 19(3): 881-904, 2017
	
	\bibitem{KNW16}
	F.\ Krahmer, D.\ Needell and R.\ Ward. \emph{Compressive sensing with redundant dictionaries and structured measurements.} SIAM J. Math. Anal. 47 (2015), no. 6: 4606--4629.
	
	\bibitem{MRW18}
	S. Mendelson,  H. Rauhut, and R. Ward. \emph{Improved bounds for sparse recovery from subsampled random convolutions.} The Annals of Applied Probability 28.6 (2018): 3491-3527.
	
	\bibitem{NT08}
	D.\ Needell and J.\ Tropp. \emph{CoSaMP: Iterative signal recovery from incomplete and inaccurate samples.} Appl. Comput. Harmon. Anal. 26 (2009), no. 3: 301--321.
	
	\bibitem{PVY17}
	Y.\ Plan, R.\ Vershynin, and E.\ Yudovina. \emph{High-dimensional estimation with geometric constraints.} Inf. Inference 6 (2017), no. 1: 1--40.
	
	\bibitem{RSV08}
	H.\ Rauhut, K.\ Schnass, and P.\ Vandergheynst. \emph{Compressed sensing and redundant dictionaries.} IEEE Trans. Inform. Theory 54 (2008), no. 5: 2210--2219.
	
	
	\bibitem{RV08}
	M.\ Rudelson and R.\ Vershynin. \emph{On sparse reconstruction from Fourier and Gaussian measurements.} Comm. Pure Appl. Math. 61 (2008), no. 8: 1025--1045. 
	
	
	
	\bibitem{Ta14}
	M.\ Talagrand. ``Upper and lower bounds for stochastic processes: modern methods and classical problems.'' Springer Science \& Business Media (2014).
	
	\bibitem{T96}
	R.\ Tibshirani. \emph{Regression shrinkage and selection via the lasso.} J. Roy. Statist. Soc. Ser. B 58 (1996), no. 1: 267--288.
	
	\bibitem{T14}
	J.\ A.\ Tropp. \emph{Convex recovery of a structured signal from independent random linear measurements.} ``Sampling theory, a renaissance.'' Appl. Numer. Harmon. Anal., Birkhäuser/Springer, Cham, (2015): 67--101.
	
	\bibitem{V15}
	R.\ Vershynin. \emph{Estimation in high dimensions: a geometric perspective.} ``Sampling theory, a renaissance.'' Appl. Numer. Harmon. Anal., Birkhäuser/Springer, Cham, (2015): 3--66.
	
	\bibitem{W15}
	M. Wainwright.
	\emph{Statistics meets optimization: randomization and approximation for high-dimensional problems}. Nachdiplom Lecture, ETH Zurich (2015). \url{https://www.stat.berkeley.edu/~wainwrig/nachdiplom/}. Scribed by Y. Wei.
	
	\bibitem{Z11}
	T. Zhang. \emph{Sparse recovery with orthogonal matching pursuit under RIP.} IEEE Trans. Inform. Theory. 57 (2011), no. 9: 6215--6221.
\end{thebibliography}
\end{document}